\documentclass[journal]{IEEEtran}
\usepackage{CJKutf8}
\usepackage{amsmath}
\usepackage{amsthm}
\newtheorem{lemma}{Lemma}

\begin{document}

\title{On $w$-Optimization of the Split Covariance Intersection Filter}

\author{Hao~Li\\
--------- \\
\small{This preprint note is extracted from the officially published book \cite{li2022FARET} written by the author.}
\thanks{H. Li, Assoc. Prof., is with Dept. Automation and SPEIT, Shanghai Jiao Tong University (SJTU), Shanghai, 200240, China (e-mail: haoli@sjtu.edu.cn)}
}

\maketitle

\begin{abstract}
The split covariance intersection filter (split CIF) is a useful tool for general data fusion and has the potential to be applied in a variety of engineering tasks. An indispensable optimization step (referred to as \textit{w}-optimization) involved in the split CIF concerns the performance and implementation efficiency of the Split CIF, but explanation on \textit{w}-optimization is neglected in the paper \cite{Li2013a} that provides a theoretical foundation for the Split CIF. This note complements \cite{Li2013a} by providing a theoretical proof for the convexity of the \textit{w}-optimization problem involved in the split CIF (convexity is always a desired property for optimization problems as it facilitates optimization considerably).
\end{abstract}

%\begin{IEEEkeywords}
%Split covariance intersection filter (Split CIF), estimation, data fusion, cooperative intelligent systems.
%\end{IEEEkeywords}

\IEEEpeerreviewmaketitle

\section{Introduction}
The paper \cite{Li2013a} provides a theoretical foundation for the split covariance intersection filter (split CIF). A reference closely related to \cite{Li2013a} is \cite{Julier2001} which presents the Split CIF heuristically without theoretical analysis --- \cite{Julier2001} originally coined it simply as ``split covariance intersection''. In \cite{Li2013a}, the term ``filter'' is added to form an analogy of the Split CIF to the well-known Kalman filter. Although the Split CIF is called ``filter'', it is not limited to temporal recursive estimation but can be used as a pure data fusion method besides the filtering sense, just as the Kalman filter can also be treated as a data fusion method --- The split CIF can reasonably handle both known independent information and unknown correlated information in source data; it is a useful tool for general data fusion and has the potential to be applied in a variety of engineering tasks \cite{Li2013b} \cite{Wanasinghe2014} \cite{Pierre2018} \cite{Chen2020} \cite{Allig2021}. 

An indispensable optimization step (referred to as \textit{w}-optimization) involved in the split CIF concerns the performance and implementation efficiency of the Split CIF; however, explanation on this \textit{w}-optimization problem is neglected in \cite{Li2013a}. As a consequence, readers may find it difficult to follow the split CIF completely as they are not informed of how the \textit{w}-optimization problem can be handled or whether the \textit{w}-optimization problem satisfies certain property (especially convexity) that facilitates optimization. To enable readers to better follow the split CIF and incorporate it into their prospective research works, this note complements \cite{Li2013a} by providing a theoretical proof for the convexity of the \textit{w}-optimization problem involved in the split CIF (convexity is always a desired property for optimization problems as it facilitates optimization considerably).

\section{The \textit{w}-optimization problem} \label{sec:woptprob}

Matrices mentioned in this note are symmetric matrices by default. Given matrices $\mathbf{P}_{1d}$, $\mathbf{P}_{1i}$, $\mathbf{P}_{2d}$, and $\mathbf{P}_{2i}$ that are positive semi-definite, i.e. $\mathbf{P}_{1d} \geq \mathbf{0}$, $\mathbf{P}_{1i} \geq \mathbf{0}$, $\mathbf{P}_{2d} \geq \mathbf{0}$, $\mathbf{P}_{2i} \geq \mathbf{0}$; denotations $\mathbf{P}_{1d}$, $\mathbf{P}_{1i}$, $\mathbf{P}_{2d}$, and $\mathbf{P}_{2i}$ are used for presentation of the Split CIF in \cite{Li2013a}. For $w \in [0,1]$, define
\begin{align} \label{eq:defineP}
\mathbf{P}_{1}(w) &= \mathbf{P}_{1d}/w + \mathbf{P}_{1i} \nonumber \\
\mathbf{P}_{2}(w) &= \mathbf{P}_{2d}/(1-w) + \mathbf{P}_{2i} \nonumber \\
\mathbf{P}(w) &= (\mathbf{P}_{1}(w)^{-1} + \mathbf{P}_{2}(w)^{-1})^{-1}
\end{align}
When $w=0$ or $w=1$, $\mathbf{P}(w)$ denotes the limit value as $w \to 0$ or $w \to 1$ respectively. For $w \in (0,1)$, we further assume that $\mathbf{P}_{1}(w)$ and $\mathbf{P}_{2}(w)$ are positive definite i.e. $\mathbf{P}_{1}(w)>0$, $\mathbf{P}_{2}(w)>0$; in fact, this fair assumption is well rooted in real applications where $\mathbf{P}_{1}(w)$ and $\mathbf{P}_{2}(w)$ normally correspond to covariances of certain estimates and hence are always positive definite. With this assumption, we naturally have $\mathbf{P}(w)>0$.

The \textit{w}-optimization problem involved in the split CIF \cite{Li2013a} can be formalized as follows:

\begin{equation}  \label{eq:wopt}
w = \arg \min_{w \in [0,1]} \det(\mathbf{P}(w))
\end{equation}

\section{Convexity of the \textit{w}-optimization problem}

We provide a theoretical proof for the convexity of the \textit{w}-optimization problem formalized in the previous section. This is equivalent to proving that the second-order differential of $\det(\mathbf{P}(w))$ in (\ref{eq:wopt}) is always non-negative for $w \in (0,1)$:
\begin{equation} \label{eq:convexC1}
\frac{d^2}{dw^2} \det(\mathbf{P}(w)) \geq 0
\end{equation} 
Note that
\begin{align*}
&\frac{d^2}{dw^2} \ln \det(\mathbf{P}(w)) \\
&= \frac{\det(\mathbf{P}(w)) \frac{d^2}{dw^2} \det(\mathbf{P}(w))-(\frac{d}{dw} \det(\mathbf{P}(w)))^2}{\det(\mathbf{P}(w))^2} \\
&\leq \frac{\frac{d^2}{dw^2} \det(\mathbf{P}(w))}{\det(\mathbf{P}(w))}
\end{align*}
So if the following inequality (\ref{eq:convexC2}) is proved, then (\ref{eq:convexC1}) holds true as well.
\begin{equation} \label{eq:convexC2}
\frac{d^2}{dw^2} \ln \det(\mathbf{P}(w)) \geq 0
\end{equation}

A detailed theoretical proof for (\ref{eq:convexC2}) is given below. For denotation conciseness in the following proof, we omit explicit writing of ``$(w)$'' for $w$-parameterized variables; for example, we denote above mentioned $\mathbf{P}_{1}(w)$, $\mathbf{P}_{2}(w)$, and $\mathbf{P}(w)$ simply as $\mathbf{P}_{1}$, $\mathbf{P}_{2}$, and $\mathbf{P}$.

\begin{lemma} \label{lm:diff1}
Given a first-order differentiable $w$-parameterized matrix $\mathbf{M}(w)$ (denoted shortly as $\mathbf{M}$) satisfying $\mathbf{M}(w)>0$, we have
\begin{equation*}
\frac{d}{dw} \ln \det(\mathbf{M}) = tr\{\mathbf{M}^{-1} \frac{d \mathbf{M}}{dw} \}
\end{equation*}
\end{lemma}

\begin{proof}
According to the Jacobi's formula \cite{Horn1991}
\begin{equation*}
\frac{d}{dw} \det(\mathbf{M}) = \det(\mathbf{M}) tr\{\mathbf{M}^{-1} \frac{d \mathbf{M}}{dw} \}
\end{equation*}
Thus we have
\begin{equation*}
\frac{d}{dw} \ln \det(\mathbf{M}) = \frac{1}{\det(\mathbf{M})} \frac{d}{dw} \det(\mathbf{M}) = tr\{\mathbf{M}^{-1} \frac{d \mathbf{M}}{dw} \}
\end{equation*}
\end{proof}

\begin{lemma} \label{lm:diff2}
Given a second-order differentiable matrix $\mathbf{M}(w)$ satisfying $\mathbf{M}(w)>0$, we have
\begin{equation*}
\frac{d^2}{dw^2} \ln \det(\mathbf{M}) = tr\{-\mathbf{M}^{-1} \frac{d \mathbf{M}}{dw} \mathbf{M}^{-1} \frac{d \mathbf{M}}{dw} + \mathbf{M}^{-1} \frac{d^2 \mathbf{M}}{dw^2} \}
\end{equation*}
\end{lemma}

\begin{proof}
Note that the differential of a matrix inverse can be computed as follows \cite{Horn1991}:
\begin{equation*}
\frac{d \mathbf{M}^{-1}}{dw} = -\mathbf{M}^{-1} \frac{d \mathbf{M}}{dw} \mathbf{M}^{-1}
\end{equation*}
Following \textbf{Lemma}.\ref{lm:diff1} we have
\begin{align*}
&\frac{d^2}{dw^2} \ln \det(\mathbf{M}) = \frac{d}{dw} tr\{\mathbf{M}^{-1} \frac{d \mathbf{M}}{dw} \}
= tr\{\frac{d}{dw} (\mathbf{M}^{-1} \frac{d \mathbf{M}}{dw}) \} \\
&= tr\{\frac{d \mathbf{M}^{-1}}{dw} \frac{d \mathbf{M}}{dw} + \mathbf{M}^{-1} \frac{d^2 \mathbf{M}}{dw^2} \} \\
&= tr\{-\mathbf{M}^{-1} \frac{d \mathbf{M}}{dw} \mathbf{M}^{-1} \frac{d \mathbf{M}}{dw} + \mathbf{M}^{-1} \frac{d^2 \mathbf{M}}{dw^2} \}
\end{align*}
\end{proof}

Following \textbf{Lemma}.\ref{lm:diff2} we can compute the second-order differential of $\ln \det(\mathbf{P}(w))$ as follows
\begin{align} \label{eq:diflndetP}
&\frac{d^2}{dw^2} \ln \det \mathbf{P} = \frac{d^2}{dw^2} \ln \det ((\mathbf{P}_1^{-1} + \mathbf{P}_2^{-1})^{-1}) \nonumber \\ 
&= \frac{d^2}{dw^2} \ln \det \mathbf{P}_1 + \frac{d^2}{dw^2} \ln \det \mathbf{P}_2 - \frac{d^2}{dw^2} \ln \det (\mathbf{P}_1 + \mathbf{P}_2) \nonumber \\ 
&= tr\{-\mathbf{P}_1^{-1} \frac{d \mathbf{P}_1}{dw} \mathbf{P}_1^{-1} \frac{d \mathbf{P}_1}{dw} + \mathbf{P}_1^{-1} \frac{d^2 \mathbf{P}_1}{dw^2} \} \nonumber \\
&+ tr\{-\mathbf{P}_2^{-1} \frac{d \mathbf{P}_2}{dw} \mathbf{P}_2^{-1} \frac{d \mathbf{P}_2}{dw} + \mathbf{P}_2^{-1} \frac{d^2 \mathbf{P}_2}{dw^2} \} \nonumber \\
&-tr\{-(\mathbf{P}_1+\mathbf{P}_2)^{-1} \frac{d (\mathbf{P}_1+\mathbf{P}_2)}{dw} (\mathbf{P}_1+\mathbf{P}_2)^{-1} \frac{d (\mathbf{P}_1+\mathbf{P}_2)}{dw} \nonumber \\
&~~~~~+ (\mathbf{P}_1+\mathbf{P}_2)^{-1} \frac{d^2 (\mathbf{P}_1+\mathbf{P}_2)}{dw^2} \} 
\end{align}

\begin{lemma} \label{lm:trcyclic}
Given two matrices $\mathbf{M}_1$ and $\mathbf{M}_2$ whose dimensions are consistent with each other for multiplication $\mathbf{M}_1 \mathbf{M}_2$ and $\mathbf{M}_2 \mathbf{M}_1$, we have $tr\{\mathbf{M}_1 \mathbf{M}_2\} = tr\{\mathbf{M}_2 \mathbf{M}_1\}$.
\end{lemma}
The proof for \textbf{Lemma}.\ref{lm:trcyclic} can be found in \cite{Horn1990}. More generally, given matrices $\mathbf{M}_1$, $\mathbf{M}_2$, and $\mathbf{M}_k$, we have 
\begin{align*}
&tr\{\mathbf{M}_1 \mathbf{M}_2 ... \mathbf{M}_k\} = tr\{\mathbf{M}_2 \mathbf{M}_3 ... \mathbf{M}_k \mathbf{M}_1\} \\
&~~~~= ... = tr\{\mathbf{M}_k \mathbf{M}_1 ... \mathbf{M}_{k-2} \mathbf{M}_{k-1}\}
\end{align*}
which is called \textit{cyclic property} of trace operation.

Define $\mathbf{D}_1 (w) = \mathbf{P}_{1d}/w$ and $\mathbf{D}_2 (w) = \mathbf{P}_{2d}/(1-w)$ for $w \in (0,1)$. As $\mathbf{P}_{1d} \geq 0$ and $\mathbf{P}_{2d} \geq 0$, we also have $\mathbf{D}_1 \geq 0$, $\mathbf{D}_2 \geq 0$. Like $\mathbf{P}_{1d}$ and $\mathbf{P}_{2d}$, $\mathbf{D}_1$ and $\mathbf{D}_2$ are also symmetric matrices. From definitions given in (\ref{eq:defineP}) we have
\begin{align*}
&\frac{d \mathbf{P}_1}{dw} = -\frac{\mathbf{D}_1}{w}~~~~~~\frac{d \mathbf{P}_2}{dw} = \frac{\mathbf{D}_2}{1-w} \\
&\frac{d^2 \mathbf{P}_1}{dw^2} = \frac{2 \mathbf{D}_1}{w^2}~~~~~~\frac{d^2 \mathbf{P}_2}{dw^2} = \frac{2 \mathbf{D}_2}{(1-w)^2}
\end{align*}
Substitute above formulas into (\ref{eq:diflndetP}) and use \textbf{Lemma}.\ref{lm:trcyclic} (the cyclic property of trace operation) when necessary in following derivation, we have
\begin{align} \label{eq:diflndetP-2}
&\frac{d^2}{dw^2} \ln \det \mathbf{P} = tr\{-\mathbf{P}_1^{-1} (-\frac{\mathbf{D}_1}{w}) \mathbf{P}_1^{-1} (-\frac{\mathbf{D}_1}{w}) + \mathbf{P}_1^{-1} \frac{2 \mathbf{D}_1}{w^2} \nonumber \\
&~-\mathbf{P}_2^{-1} (\frac{\mathbf{D}_2}{1-w}) \mathbf{P}_2^{-1} (\frac{\mathbf{D}_2}{1-w}) + \mathbf{P}_2^{-1} \frac{2 \mathbf{D}_2}{(1-w)^2} \nonumber \\
&~+(\mathbf{P}_1+\mathbf{P}_2)^{-1} (\frac{\mathbf{D}_2}{1-w}-\frac{\mathbf{D}_1}{w}) (\mathbf{P}_1+\mathbf{P}_2)^{-1} (\frac{\mathbf{D}_2}{1-w}-\frac{\mathbf{D}_1}{w}) \nonumber \\
&~- (\mathbf{P}_1+\mathbf{P}_2)^{-1} (\frac{2 \mathbf{D}_1}{w^2}+\frac{2 \mathbf{D}_2}{(1-w)^2}) \} \nonumber \\
&= \frac{1}{w^2} \mathbf{T}_1 + \frac{1}{(1-w)^2} \mathbf{T}_2 - \frac{2}{w(1-w)} \mathbf{T}_3
\end{align}
where
\begin{align*}
&\mathbf{T}_1 = tr\{2 \mathbf{P}_1^{-1} \mathbf{D}_1 - 2 (\mathbf{P}_1+\mathbf{P}_2)^{-1} \mathbf{D}_1 - \mathbf{P}_1^{-1} \mathbf{D}_1 \mathbf{P}_1^{-1} \mathbf{D}_1 \\
&~~~~~~+ (\mathbf{P}_1+\mathbf{P}_2)^{-1} \mathbf{D}_1 (\mathbf{P}_1+\mathbf{P}_2)^{-1} \mathbf{D}_1 \} \\
&\mathbf{T}_2 = tr\{2 \mathbf{P}_2^{-1} \mathbf{D}_2 - 2 (\mathbf{P}_1+\mathbf{P}_2)^{-1} \mathbf{D}_2 - \mathbf{P}_2^{-1} \mathbf{D}_2 \mathbf{P}_2^{-1} \mathbf{D}_2 \\
&~~~~~~+ (\mathbf{P}_1+\mathbf{P}_2)^{-1} \mathbf{D}_2 (\mathbf{P}_1+\mathbf{P}_2)^{-1} \mathbf{D}_2 \} \\
&\mathbf{T}_3 = tr\{(\mathbf{P}_1+\mathbf{P}_2)^{-1} \mathbf{D}_1 (\mathbf{P}_1+\mathbf{P}_2)^{-1} \mathbf{D}_2\}
\end{align*}

\begin{lemma} \label{lm:ineq1}
Given two positive semi-definite matrices $\mathbf{M}_1$ and $\mathbf{M}_2$ (i.e. $\mathbf{M}_1 \geq 0$, $\mathbf{M}_2 \geq 0$), we have $tr\{\mathbf{M}_1 \mathbf{M}_2\} = tr\{\mathbf{M}_2 \mathbf{M}_1\} \geq 0$.
\end{lemma}
The proof for \textbf{Lemma}.\ref{lm:ineq1} can be found in \cite{Horn1990}.

\begin{lemma} \label{lm:ineq2}
Given symmetric matrices $\mathbf{X}$, $\mathbf{Y}$, and $\mathbf{Z}$ satisfying $0 < \mathbf{X} \leq \mathbf{Y}$ and $0 \leq \mathbf{Z} \leq \mathbf{X}$, we have
\begin{align*}
&tr\{ 2 \mathbf{X}^{-1} \mathbf{Z} - 2 \mathbf{Y}^{-1} \mathbf{Z} - \mathbf{X}^{-1} \mathbf{Z} \mathbf{X}^{-1} \mathbf{Z} + \mathbf{Y}^{-1} \mathbf{Z} \mathbf{Y}^{-1} \mathbf{Z} \} \\
&~~~~\geq tr\{ (\mathbf{X}^{-1}-\mathbf{Y}^{-1}) \mathbf{Z} (\mathbf{X}^{-1}-\mathbf{Y}^{-1}) \mathbf{Z} \}
\end{align*} 
\end{lemma}

\begin{proof}
\textbf{Lemma}.\ref{lm:trcyclic} is used in following derivation
\begin{align*}
&tr\{ 2 \mathbf{X}^{-1} \mathbf{Z} - 2 \mathbf{Y}^{-1} \mathbf{Z} - \mathbf{X}^{-1} \mathbf{Z} \mathbf{X}^{-1} \mathbf{Z} + \mathbf{Y}^{-1} \mathbf{Z} \mathbf{Y}^{-1} \mathbf{Z} \} \\
&~~~~- tr\{ (\mathbf{X}^{-1}-\mathbf{Y}^{-1}) \mathbf{Z} (\mathbf{X}^{-1}-\mathbf{Y}^{-1}) \mathbf{Z} \} \\
&= tr\{ 2 \mathbf{X}^{-1} \mathbf{Z} - 2 \mathbf{Y}^{-1} \mathbf{Z} - 2 \mathbf{X}^{-1} \mathbf{Z} \mathbf{X}^{-1} \mathbf{Z} \\
&\quad\qquad\qquad\qquad+ \mathbf{X}^{-1} \mathbf{Z} \mathbf{Y}^{-1} \mathbf{Z} + \mathbf{Y}^{-1} \mathbf{Z} \mathbf{X}^{-1} \mathbf{Z} \} \\
&= tr\{ 2 \mathbf{X}^{-1} \mathbf{Z} - 2 \mathbf{Y}^{-1} \mathbf{Z} - 2 \mathbf{X}^{-1} \mathbf{Z} \mathbf{X}^{-1} \mathbf{Z} + 2 \mathbf{X}^{-1} \mathbf{Z} \mathbf{Y}^{-1} \mathbf{Z} \} \\
&= 2~tr\{ (\mathbf{I}-\mathbf{X}^{-1} \mathbf{Z}) (\mathbf{X}^{-1}-\mathbf{Y}^{-1}) \mathbf{Z} \} \\
&= 2~tr\{\mathbf{Z} (\mathbf{I}-\mathbf{X}^{-1} \mathbf{Z}) (\mathbf{X}^{-1}-\mathbf{Y}^{-1} ) \} \\
&= 2~tr\{\mathbf{Z} (\mathbf{Z}^{-1}-\mathbf{X}^{-1}) \mathbf{Z} (\mathbf{X}^{-1}-\mathbf{Y}^{-1}) \}
\end{align*}
As $\mathbf{Z}^{-1}-\mathbf{X}^{-1} \geq 0$, we have
\begin{align*}
\mathbf{Z} (\mathbf{Z}^{-1}-\mathbf{X}^{-1}) \mathbf{Z} = \mathbf{Z}^T (\mathbf{Z}^{-1}-\mathbf{X}^{-1}) \mathbf{Z} \geq 0
\end{align*}
Besides, as $\mathbf{X}^{-1}-\mathbf{Y}^{-1} \geq 0$; following \textbf{Lemma}.\ref{lm:ineq1} we have $tr\{\mathbf{Z} (\mathbf{Z}^{-1}-\mathbf{X}^{-1}) \mathbf{Z} (\mathbf{X}^{-1}-\mathbf{Y}^{-1}) \} \geq 0$. The proof is done
\end{proof}

Note that $\mathbf{P}_1$, $\mathbf{P}_2$, $\mathbf{D}_1$, $\mathbf{D}_2$, and $\mathbf{P}_1+\mathbf{P}_2$ are symmetric matrices satisfying $\mathbf{P}_1+\mathbf{P}_2 > \mathbf{P}_1 = \mathbf{D}_1+\mathbf{P}_{1i} \geq \mathbf{D}_1 \geq 0$ and $\mathbf{P}_1+\mathbf{P}_2 > \mathbf{P}_2 = \mathbf{D}_2+\mathbf{P}_{2i} \geq \mathbf{D}_2 \geq 0$; following \textbf{Lemma}.\ref{lm:ineq2} we have (denote $\mathbf{P}_3 = \mathbf{P}_1+\mathbf{P}_2$)
\begin{align*}
\mathbf{T}_1 \geq tr\{ (\mathbf{P}_1^{-1}-\mathbf{P}_3^{-1}) \mathbf{D}_1 (\mathbf{P}_1^{-1}-\mathbf{P}_3^{-1}) \mathbf{D}_1 \} \\
\mathbf{T}_2 \geq tr\{ (\mathbf{P}_2^{-1}-\mathbf{P}_3^{-1}) \mathbf{D}_2 (\mathbf{P}_2^{-1}-\mathbf{P}_3^{-1}) \mathbf{D}_2 \}
\end{align*}

Substitute above inequalities into (\ref{eq:diflndetP-2}) and we have 
\begin{align} \label{eq:lndetPineq}
&\frac{d^2}{dw^2} \ln \det \mathbf{P} \geq tr\{ (\mathbf{P}_1^{-1}-\mathbf{P}_3^{-1}) \frac{\mathbf{D}_1}{w} (\mathbf{P}_1^{-1}-\mathbf{P}_3^{-1}) \frac{\mathbf{D}_1}{w} \} \nonumber \\
&\quad\qquad + tr\{ (\mathbf{P}_2^{-1}-\mathbf{P}_3^{-1}) \frac{\mathbf{D}_2}{1-w} (\mathbf{P}_2^{-1}-\mathbf{P}_3^{-1}) \frac{\mathbf{D}_2}{1-w} \} \nonumber \\
&\quad\qquad - 2~tr\{\mathbf{P}_3^{-1} \frac{\mathbf{D}_1}{w} \mathbf{P}_3^{-1} \frac{\mathbf{D}_2}{1-w}\}
\end{align}
Denote $\mathbf{B}_3 = \mathbf{P}_1^{-1}+\mathbf{P}_2^{-1}$. Note that
\begin{align*}
\mathbf{P}_3^{-1} &= (\mathbf{P}_1+\mathbf{P}_2)^{-1} = (\mathbf{P}_1 (\mathbf{P}_1^{-1}+\mathbf{P}_2^{-1}) \mathbf{P}_2)^{-1} \\
&= \mathbf{P}_2^{-1} (\mathbf{P}_1^{-1}+\mathbf{P}_2^{-1})^{-1} \mathbf{P}_1^{-1} \\
&= \mathbf{P}_2^{-1} \mathbf{B}_3^{-1} \mathbf{P}_1^{-1} \\
\textnormal{or} \quad \mathbf{P}_3^{-1} &= (\mathbf{P}_2 (\mathbf{P}_1^{-1}+\mathbf{P}_2^{-1}) \mathbf{P}_1)^{-1} = \mathbf{P}_1^{-1} \mathbf{B}_3^{-1} \mathbf{P}_2^{-1}
\end{align*}
We have
\begin{align*}
\mathbf{P}_1^{-1}-\mathbf{P}_3^{-1} &= \mathbf{P}_1^{-1} - \mathbf{P}_2^{-1} (\mathbf{P}_1^{-1}+\mathbf{P}_2^{-1})^{-1} \mathbf{P}_1^{-1} \\
&= ((\mathbf{P}_1^{-1}+\mathbf{P}_2^{-1})  - \mathbf{P}_2^{-1}) (\mathbf{P}_1^{-1}+\mathbf{P}_2^{-1})^{-1} \mathbf{P}_1^{-1} \\
&= \mathbf{P}_1^{-1} (\mathbf{P}_1^{-1}+\mathbf{P}_2^{-1})^{-1} \mathbf{P}_1^{-1} \\
&= \mathbf{P}_1^{-1} \mathbf{B}_3^{-1} \mathbf{P}_1^{-1}
\end{align*}
Similarly we have
\begin{align*}
\mathbf{P}_2^{-1}-\mathbf{P}_3^{-1} = \mathbf{P}_2^{-1} \mathbf{B}_3^{-1} \mathbf{P}_2^{-1}
\end{align*}
Therefore, (\ref{eq:lndetPineq}) becomes
\begin{align} \label{eq:lndetPineq-2}
&\frac{d^2}{dw^2} \ln \det \mathbf{P} \nonumber \\
&\quad \geq tr\{ \mathbf{P}_1^{-1} \mathbf{B}_3^{-1} \mathbf{P}_1^{-1} \frac{\mathbf{D}_1}{w} \mathbf{P}_1^{-1} \mathbf{B}_3^{-1} \mathbf{P}_1^{-1} \frac{\mathbf{D}_1}{w} \} \nonumber \nonumber \\
&\qquad + tr\{ \mathbf{P}_2^{-1} \mathbf{B}_3^{-1} \mathbf{P}_2^{-1} \frac{\mathbf{D}_2}{1-w} \mathbf{P}_2^{-1} \mathbf{B}_3^{-1} \mathbf{P}_2^{-1} \frac{\mathbf{D}_2}{1-w} \} \nonumber \nonumber \\
&\qquad - 2~tr\{\mathbf{P}_2^{-1} \mathbf{B}_3^{-1} \mathbf{P}_1^{-1} \frac{\mathbf{D}_1}{w} \mathbf{P}_1^{-1} \mathbf{B}_3^{-1} \mathbf{P}_2^{-1} \frac{\mathbf{D}_2}{1-w}\} \nonumber \\
&\quad = tr\{ \mathbf{B}_3^{-1} \mathbf{P}_1^{-1} \frac{\mathbf{D}_1}{w} \mathbf{P}_1^{-1} \mathbf{B}_3^{-1} \mathbf{P}_1^{-1} \frac{\mathbf{D}_1}{w} \mathbf{P}_1^{-1} \} \nonumber \nonumber \\
&\qquad + tr\{ \mathbf{B}_3^{-1} \mathbf{P}_2^{-1} \frac{\mathbf{D}_2}{1-w} \mathbf{P}_2^{-1} \mathbf{B}_3^{-1} \mathbf{P}_2^{-1} \frac{\mathbf{D}_2}{1-w} \mathbf{P}_2^{-1} \} \nonumber \nonumber \\
&\qquad - 2~tr\{ \mathbf{B}_3^{-1} \mathbf{P}_1^{-1} \frac{\mathbf{D}_1}{w} \mathbf{P}_1^{-1} \mathbf{B}_3^{-1} \mathbf{P}_2^{-1} \frac{\mathbf{D}_2}{1-w} \mathbf{P}_2^{-1} \} \nonumber \\
&\quad = tr\{ \mathbf{B}_3^{-1} \mathbf{C} \mathbf{B}_3^{-1} \mathbf{C} \}
\end{align}
where
\begin{align*}
\mathbf{C} = \mathbf{P}_1^{-1} \frac{\mathbf{D}_1}{w} \mathbf{P}_1^{-1} - \mathbf{P}_2^{-1} \frac{\mathbf{D}_2}{1-w} \mathbf{P}_2^{-1}
\end{align*}
As matrices $\mathbf{P}_1$, $\mathbf{P}_2$, $\mathbf{D}_1$, and $\mathbf{D}_2$ are all symmetric, so is $\mathbf{C}$. Note that $\mathbf{B}_3 = \mathbf{P}_1^{-1}+\mathbf{P}_2^{-1} > 0$ ($\mathbf{B}_3$ is symmetric as well) and hence $\mathbf{B}_3^{-1} > 0$, we have
\begin{align*}
\mathbf{C} \mathbf{B}_3^{-1} \mathbf{C} = \mathbf{C}^T \mathbf{B}_3^{-1} \mathbf{C} \geq 0
\end{align*}
Follow (\ref{eq:lndetPineq-2}) and \textbf{Lemma}.\ref{lm:ineq1} and we have
\begin{align*}
\frac{d^2}{dw^2} \ln \det \mathbf{P} \geq tr\{ \mathbf{B}_3^{-1} \mathbf{C} \mathbf{B}_3^{-1} \mathbf{C} \} \geq 0
\end{align*}
So all the proof for (\ref{eq:convexC2}) is presented. As we have already explained at the beginning of this section, (\ref{eq:convexC1}) also holds true and the convexity of the $w$-optimization problem is proved.

\section{Conclusion}

Explanation on an indispensable optimization step (i.e. the $w$-optimization problem) involved in the split CIF is neglected in \cite{Li2013a}, this note complements \cite{Li2013a} by providing a theoretical proof with details for the convexity of the $w$-optimization problem. As convexity facilitates optimization considerably, readers can resort to convex optimization techniques to solve the $w$-optimization problem when they intend to incorporate the split CIF into their prospective research works.

\section*{Appendix}

Demo code: https://github.com/LI-Hao-SJTU/SplitCIF 

\ifCLASSOPTIONcaptionsoff
  \newpage
\fi
\bibliographystyle{IEEEtran}
\bibliography{LI_Ref}

\end{document}